\documentclass[a4paper,10pt]{article}
\usepackage{authblk}

\usepackage[ngerman,english]{babel}
\usepackage[latin1]{inputenc}
\usepackage{csquotes}

\usepackage{amsfonts,amsmath,amsthm}
\usepackage{empheq}
\usepackage{enumitem}

\usepackage[titletoc,title]{appendix}

\usepackage[backend=bibtex8,url=true,doi=false,eprint=false,giveninits=true,isbn=false,style=numeric-comp,maxnames=99]{biblatex}
\makeatletter
\def\blx@maxline{77}
\makeatother

\bibliography{bibl.bib}
\AtEveryBibitem{\clearlist{language}}

\usepackage{cases}
\usepackage{mathabx}
\usepackage{bbm}
\usepackage{xfrac}
\usepackage{slashbox}
\usepackage{adjustbox}
\usepackage{overpic}
\usepackage{fp}

\usepackage{caption}
\usepackage
{subcaption}

\usepackage{xstring}

\usepackage{fancyhdr}
\usepackage{pgf,pgffor}

\usepackage{color}
\usepackage[colorinlistoftodos,textsize=small,backgroundcolor=white,bordercolor=magenta,linecolor=magenta]{todonotes}
\usepackage[colorlinks]{hyperref}
\definecolor{blue75}{rgb}{0,0,.75}
\definecolor{green75}{rgb}{0,.75,0}
\hypersetup{colorlinks=true, urlcolor=blue75,linkcolor=blue75,citecolor=green75,pdfstartview=FitB,bookmarksopen=true,bookmarksopenlevel=1}

\usepackage[a4paper, left=2.5cm, right=2.5cm, top=2.5cm,bottom=2cm]{geometry}
\usepackage{constants}
\newcommand{\parenthezises}[1]{\arabic{#1}}
\newconstantfamily{C}{
symbol=C,
format=\parenthezises,
reset={section}
}
\newconstantfamily{G}{
symbol=G,
format=\parenthezises,
reset={section}
}

\newconstantfamily{M}{
symbol=M,
format=\parenthezises,
reset={section}
}
\newconstantfamily{B}{
symbol=B,
format=\parenthezises,
reset={section}
}
\newconstantfamily{xi}{
symbol=\xi,
format=\parenthezises,
reset={section}
}
\usepackage{graphicx}
\usepackage{wrapfig}
\usepackage{figbib}
\allowdisplaybreaks
\usepackage[capitalise]{cleveref}

\crefdefaultlabelformat{{\it #2#1#3}}

\crefname{equation}{}{}

\crefname{enumi}{}{}
\crefformat{enumi}{#2#1#3}

\crefname{section}{{\it Section}}{{\it Sections}}
\crefname{subsection}{{\it Subsection}}{{\it Subsections}}
\crefname{subsubsection}{{\it Paragraph}}{{\it Paragraphs}}
\crefname{table}{{\it Table}}{{\it Tables}}

\crefname{figure}{{\it Figure}}{{\it Figures}}
\newtheorem{Theorem}{Theorem}[section]
\crefname{Theorem}{{\it Theorem}}{{\it Theorems}}

\crefname{Definition}{{\it Definition}}{{\it Definitions}}
\newtheorem{Lemma}[Theorem]{Lemma}
\crefname{Lemma}{{\it Lemma}}{{\it Lemmas}}

\crefname{Proposition}{{\it Proposition}}{{\it Propositions}}
\newtheorem{Assumptions}[Theorem]{Assumptions}
\crefname{Assumptions}{{\it Assumptions}}{{\it Assumptions}}

\theoremstyle{definition}
\newtheorem{Remark}[Theorem]{Remark}
\crefname{Remark}{{\it Remark}}{{\it Remarks}}

\crefname{Notation}{{\it Notation}}{{\it Notations}}
\newtheorem{Example}[Theorem]{Example}
\crefname{Example}{{\it Example}}{{\it Examples}}

\counterwithin{figure}{section}
\counterwithin{table}{section}

\usepackage[version=4]{mhchem}

\makeatletter
\def\@fnsymbol#1{\ensuremath{\ifcase#1\or *\or **\else\@ctrerr\fi}}
\makeatother

\begin{document}
\title{Modelling non-local cell-cell adhesion: a multiscale approach}
\author{Anna Zhigun\thanks{School of Mathematics and Physics, Queen's University Belfast, University Road, Belfast BT7 1NN, Northern Ireland, UK, \href{mailto:A.Zhigun@qub.ac.uk}{A.Zhigun@qub.ac.uk}} \ and
    Mabel Lizzy Rajendran\thanks{School of Mathematics, 
Watson Building, University of Birmingham, Edgbaston, Birmingham
B15 2TT, UK, \href{mailto:m.l.rajendran@bham.ac.uk}{m.l.rajendran@bham.ac.uk}} 
}

\date{}
\maketitle
\begin{abstract} 
Cell-cell adhesion plays a vital role in the development and maintenance of multicellular organisms. One of its functions is regulation of cell migration, such as occurs, e.g. during embryogenesis or in cancer.  
In this work, we develop a versatile multiscale approach  to modelling a moving self-adhesive cell population  that combines a careful microscopic  description of a deterministic adhesion-driven motion component with an efficient mesoscopic representation of a stochastic  velocity-jump process. This approach gives rise to mesoscopic models in the form of kinetic transport equations featuring multiple non-localities. Subsequent parabolic and hyperbolic scalings produce  general classes of equations with non-local adhesion and myopic diffusion, a special case being the classical macroscopic model proposed in  \cite{Armstrong2006}. Our simulations show how the combination of the two motion effects can unfold. 

Cell-cell adhesion relies on the subcellular cell adhesion molecule binding. 
  Our approach lends itself conveniently to capturing this microscopic effect. On the macroscale, this results in an additional non-linear  integral equation of a novel type that is coupled to the cell density equation.
\end{abstract}

\newcommand{\mres}{\mathbin{\vrule height 1.6ex depth 0pt width
0.13ex\vrule height 0.13ex depth 0pt width 1.3ex}}

\newcommand{\cb}[1]{{\color{blue}#1}}
 \newcommand{\cred}[1]{{\color{red}#1}}
\newcommand{\cmg}[1]{{\color{magenta}#1}}
\newcommand{\cgr}[1]{{\color{gray}#1}}

\newcommand{\D}{\mathbb{D}}
\newcommand{\E}{\mathbb{E}}
\newcommand{\T}{{\mathbb T}}
\newcommand{\PP}{{\cal P}}
\newcommand{\A}{{\cal A}}
\newcommand{\B}{{\cal B}}
\newcommand{\K}{{\cal K}}
\newcommand{\G}{{\cal G}}
\newcommand{\Y}{{\cal Y}}

\newcommand{\LL}{{\cal L}}
\newcommand{\M}{{\cal M}}
\newcommand{\HH}{{\cal H}}
\newcommand{\R}{\mathbb{R}}
\newcommand{\N}{\mathbb{N}}
\newcommand{\F}{\mathbb{F}}
\newcommand{\V}{\mathbb{V}}
\newcommand{\X}{{\cal X}}
\newcommand{\W}{{\cal W}}
\newcommand{\ve}{\varepsilon}
\newcommand{\cep}{c^{\ve}}
\def\diam{\operatorname{diam}}
\def\dist{\operatorname{dist}}
\def\diver{\operatorname{div}}
\def\ess{\operatorname{ess}}
\def\inner{\operatorname{int}}
\def\osc{\operatorname{osc}}
\def\sign{\operatorname{sign}}
\def\supp{\operatorname{supp}}
\newcommand{\tr}{\operatorname{tr}}
\newcommand{\Om}{{\Omega}}
\newcommand{\oOm}{\overline{\Omega}}
\newcommand{\MR}{{\cal M}(\overline{\Omega})}
\newcommand{\MP}{{\cal M}^+(\overline{\Omega})}
\newcommand{\MG}{{\cal M}_G(\overline{\Omega})}
\newcommand{\MGp}{{\cal M}_G^+(\overline{\Omega})}
\newcommand{\MO}{{\cal M}_{KR}(\overline{\Omega})}

\newcommand{\BMO}{BMO(\Omega)}

\newcommand{\LOne}{L^{1}(\Omega)}
\newcommand{\WO}{W^{1,1}(\Omega)}
\newcommand{\LOnen}{(L^{1}(\Omega))^d}
\newcommand{\LTwo}{L^{2}(\Omega)}
\newcommand{\Lq}{L^{q}(\Omega)}
\newcommand{\Lp}{L^{2}(\Omega)}
\newcommand{\Lpn}{(L^{2}(\Omega))^d}
\newcommand{\LInf}{L^{\infty}(\Omega)}
\newcommand{\WInf}{W^{1,\infty}(\Omega)}
\newcommand{\WI}{W^{1,\infty}(\Omega_0)}
\newcommand{\Ca}{C^{\alpha }(\overline{\Omega})}
\newcommand{\HOneO}{H^{1,0}(\Omega)}
\newcommand{\HTwoO}{H^{2,0}(\Omega)}
\newcommand{\HOne}{H^{1}(\Omega)}
\newcommand{\HTwo}{H^{2}(\Omega)}
\newcommand{\HmOne}{H^{-1}(\Omega)}
\newcommand{\HmTwo}{H^{-2}(\Omega)}

\newcommand{\LlogL}{L\log L(\Omega)}

\def\avint{\mathop{\,\rlap{-}\!\!\int}\nolimits} 

\newcommand{\om}{\omega}
\newcommand{\tu}{\tilde{u}}
\newcommand{\tw}{\tilde{w}}

\newtheorem{Step}{Step}
\makeatletter
\@addtoreset{Step}{Theorem}
\makeatother
\numberwithin{equation}{section}

\section{Introduction}\label{SecIntro}
\subsection{Biological background}\label{IntroBB}
Development and functioning of multicellular organisms crucially depend on cell-cell adhesion (CCA).  This is the process of cells binding to their neighbours to form  multicellular complexes by building cell-cell junctions. Formation of new tissues and organs during embryogenesis as well as their maintenance, be it as part of homeostasis or during wound healing, all rely on CCA. Alteration of CCA is linked to cancer invasion and metastasis   \cite{HanahanWeinberger2011,FriedlAlexander2011}. 

Various types of cell-cell junctions exist, each enabling a specific adhesion functionality. Mainly responsible for keeping cells together are adherens junctions \cite[Chapter 19]{Albertsetal}. They are facilitated by a particular type of cell adhesion molecule (CAM),  an adhesion-mediating transmembrane protein, called cadherin. Cadherins require extracellular calcium ($Ca^{2+}$) in order to form junctions. Through catenins, a family of intracellular proteins, cadherins are indirectly connected to actin filaments that are part of the cell cytoskeleton. Since individual cadherin bonds are rather weak, many such bonds need to be established in parallel to secure a strong  anchoring junction between two cells. A comprehensive description of these junctions can be found in \cite[Chapter 19]{Albertsetal}. 

Depending on other chemical factors, cadherins either suppress   migration out of resting epithelium or support collective invasion. While E-cadherin is responsible for strong bonds in the former case, various members of the cadherin family that have weaker adhesion strengths than E-cadherin, such as N-cadherin,
 are mainly observed when migration occurs. Loss of E-cadherin, the main cell-cell binding CAM in epithelial cells, is believed to be a fundamental event in the epithelial-mesenchymal transition (EMT), a process by which cells switch from epithelial to mesenchymal stem type. In cancer, EMT enables invasion, the precursor of metastasis. 
 Another key  adhesion type is cell-tissue adhesion, yet in this work we concentrate solely on CCA.  
 We refer to \cite{FriedlAlexander2011,HanahanWeinberger2011}, as well as references in these sources, for further  details on the role of cell adhesion in cancer.

Motivated by the biological observations outlined above, our main aim in this paper is to derive 
two prototypical classes of continuum   mathematical models for a diffusion-advection-driven motion of a self-adhesive cell population  in a heterogeneous environment.
The model in  \cite{Armstrong2006} is regained as a special case of one of these models,  \cref{PL}. This equation does not account for calcium-mediated cadherin binding dynamics, whereas our novel model \cref{MacroModel} includes such dynamics. 

\subsection{Modelling background}\label{SecIntroMod}
Reaction-diffusion-advection  equations with a density-dependent non-local advection  velocity in the form of a spatial integral 
are a popular choice when it comes to modelling adhesion on the level of population densities. Such non-local terms indirectly account for cell-cell interactions  through the effect that they have on the bulk motion.  
Starting from the integro-partial differential equation (IPDE) that was proposed in \cite{Armstrong2006} many extensions of that adhesion model  were developed and treated mathematically rigorously and numerically, see \cite{reviewNonlocal2020} and references therein. Numerical simulations confirm that models of this sort reproduce aggregate formation caused by CCA.

However, it is hardly possible to capture important information that needs to be passed from lower scales to obtain a realistic model if  modelling is done directly on the level of densities. More accurate models are obtained by zooming to the cellular or even subcellular levels and/or the level of cell density distributions and subsequently performing an upscaling. Several   approaches to such derivations were adopted in the context of IPDE (diffusion)-adhesion models.  
 We briefly review them. 
\begin{enumerate}[label=M\arabic*,ref=M\arabic*]
\item\label{SDEappr}
One possible starting point is a system of a large number of first-order stochastic differential equations (SDEs). 
Each of the SDEs describes the temporal evolution of the spatial position of a single population member  that interacts with other individuals and is influenced by stochastic fluctuations, typically in the form of a Gaussian white noise. Interaction, which is often a combination of multiple  effects, is characterised by an appropriate interaction potential. 
In the cases where one preliminarily considers a system of second order SDEs for spatial positions, simplifying assumptions are  made in order to reduce to the first-order as above.

Using empirical measures and It\^o's formula, one can deduce the corresponding mean field stochastic IPDE for population density that contains a non-locality in the drift term. Choosing suitable scalings in each part of the interaction potential allows to remove randomness: as the number of particles tends to infinity,  a fully  deterministic macroscopic mean field IPDE is obtained. The interaction potential contributes with non-local advection and/or the non-linear part of diffusion. Stochastic fluctuations produce a linear diffusion component.
For a detailed discussion of this approach to derivation of  non-local diffusion-advection models spanning from modelling aspects to rigorous mathematical treatment we refer to   \cite{MorCapOel2005,MorCapOel1998} as well as references therein. 

In \cite{Middleton}, two classes of mean field IPDE models of adhesion  were derived while keeping the number of individuals finite. This was achieved by imposing  closing relations. The mean-field approximation yielded an IPDE of the same form as in \cite{Armstrong2006}, whereas the Kirkwood superposition approximation produced a non-standard system of two strongly coupled IPDEs.

{No application of the described approach has dealt with CAM binding dynamics or comparable effects.} 
 \item\label{Masterappr} Assuming that cell motion follows a space-jump random process, the  evolution of the population density can be modelled by a Master equation. The time derivative of the density  is then given by a spatial integral operator   governed by a redistribution kernel that describes the probability of jumping from one position to another. 
Under suitable assumptions on the redistribution kernel one can rescale the equation and then pass to the limit letting the jump length tend to zero, with the result being a diffusion-advection PDE. This approach allowed to formally derive IPDE diffusion-adhesion models in  \cite{ButtenHGP2018}, see also references therein on further details regarding  derivation and scaling of Master equations. 
 
 In \cite{ButtenHGP2018}, the redistribution kernel was split into the even and odd components, leading to  a myopic diffusion and advection, respectively. The odd component was generated by the so-called cell's  polarization vector. This vector was assumed to be a superposition of  local adhesion strengths  generated in small test volumes in the cell's environment. %
 The local adhesion strength was assumed to be proportional to: the  distance to the cell, the available space, and the amount of bound CAMs. Depending on the way the latter evolution was modelled, one obtained adhesion velocities proportional to either single or double spatial integrals. A special case of the first option led to the model in  \cite{Armstrong2006}. 
 The subcellular binding-unbinding dynamics of CAMs was described by quickly  equilibrating ordinary differential equations (ODEs). Their coefficients were obtained from local densities or their integrals. 
 \item\label{KTEappr} If it can be assumed that cell motion follows a velocity-jump process, then a kinetic transport equation (KTE) lends itself to the  description of the evolution of the mesoscopic cell density with respect to time, position, velocity, and, possibly, other variables (see \cref{KTAPmethod} below). 
{In the absence of source terms, it} takes the form of an IPDE  where the differential transport term captures the  deterministic directed movement and a velocity integral term cumulates the effect of stochastic fluctuations due to switching from one velocity to another. The later term is the  turning operator based on a turning kernel that gives the probability of velocity  switches. 
 Removing the turning operator would leave us with a mesoscopic mean field equation of the form of a conservative transport equation (CTE). This equation can be obtained from a microscopic ODE system describing deterministic cell movement using empirical measures, i.e. in the same fashion as described in \cref{SDEappr}. It is commonly assumed that velocity changes are purely stochastic in nature, implying the absence of velocity-induced transport. However, a few works have introduced such a transport term into the model \cite{ZSMM,DKSS2022,CKSWNS2021,Chauviere2007}. {We mention in passing that in the physical context similar equations exist. For instance, the linear Boltzmann-Maxwell equation can describe a gas
of charged particles moving under the influence of an external field through an unchanging
background of another type of particles. Yet cell interactions with other cells and lifeless matter in their surroundings are unlike collisions of physical particles with each other or their background. Hence, different kind of forces and interaction kernels need to be considered.} 
 
 Upscaling, i.e. a suitable rescaling of a KTE using a small scaling parameter and a limit procedure while this  parameter is being sent to zero, yields a mean field equation for the macroscopic population density. Standard scalings are the hyperbolic and the parabolic ones. As a rule, the hyperbolic scaling is available without preconditions, whereas the  parabolic scaling is only possible if the first velocity moment of the turning operator tends to zero. In contrast to the parabolic upscaling where diffusion arises directly,  higher-order correction terms need to be included in order to produce diffusion in the hyperbolic scaling case. We refer to \cite{HillenPainter} for a detailed discussion of the two scaling types in the context of movement of living organisms, cells in particular.

 {This approach 
 was applied in the contexts of self-organised animal aggregations \cite{CEH2015,EftimieThesis} and cell dispersal mediated by non-local sensing \cite{LoyPrez20201,LoyPrez20202}.
 In these works, formal derivations of advection(-diffusion) equations were based on KTEs containing transport with respect to the spatial variable alone. To produce non-localities on the macroscale,  turning operators were chosen that contained both velocity and spatial averaging. In the first setting, the focus was on the interaction inside a population in one \cite{CEH2015,EftimieThesis} and two \cite{CEH2015} dimensions (see also references in these works with regard to previous mesoscopic modelling in this context). This multilayered effect was modelled by a turning kernel that is split into constants and samplings, over the whole domain, of several interaction sub-kernels. Along with density- and distance-dependent weightings,  the latter involves  functions  that measure  differences between the previous velocity direction and the directions of: the future velocity and the neighbours' velocity and relative position. 
 The modelling in  \cite{LoyPrez20201,LoyPrez20202} had no dimension restrictions and aimed at describing cell polarization. There, the turning kernel samples density- and distance-weighted measurements of a macroscopic quantity at positions along the future velocity direction. Choosing this quantity to be the cell population density, non-local CCA can be captured. 
 
 A suitable rescaling in \cite{CEH2015} kept the lowest-order part of the kernel velocity-independent, allowing for a parabolic scaling that yielded a PDE with non-localities in both the diffusion coefficient and the adhesion velocity.
 Due to the structure of the kernel, the parabolic scaling in \cite{LoyPrez20201,LoyPrez20202} was only possible under an  assumption that removed the non-locality. 
 
 CAM binding dynamics or comparable effects were not considered in either of these works.}  
\end{enumerate}

In the present work we develop an alternative approach to modelling CCA. Our main goal is to construct a flexible multiscale modelling framework that captures better the  biological observations described in \cref{IntroBB} above. To be precise, we want to ensure that on the macroscale:
\begin{enumerate}[label=(\roman*),ref=(\roman*),series=method]
 \item\label{GoalDiff} the diffusion term originates from stochastic fluctuations;
 \item the adhesion term is non-local,  
 \item\label{GoalCam}  microscopic CAM binding dynamics being its source.  
\end{enumerate}
\begin{table}[h!]
 \centering
 \begin{tabular}{l|l|l|l}
 \backslashbox{Method}{Origin of} & adhesion term&  diffusion term& CAM dynamics\\\hline
 \cref{SDEappr}&deterministic&deterministic \& stochastic&none \\\hline
 \cref{Masterappr}&stochastic \& CAMs &stochastic&deterministic \\\hline
 \cref{KTEappr}&stochastic&{stochastic}&none\\\hline
 \cref{SecCAMs}&deterministic \&  CAMs&stochastic&deterministic
\end{tabular} 
\caption{Comparison of \cref{SDEappr}-\cref{KTEappr} and our {approach}.}\label{TabComp}
\end{table}
\cref{TabComp} compares the above outlined approaches and our approach in \cref{SecCAMs} (discussed below). The comparison is based on criteria that are related to the declared objectives \cref{GoalDiff}-\cref{GoalCam}. 
 Each of \cref{SDEappr}-\cref{KTEappr} meets the requirements partially, yet fails to meet them all. Note that 
 even though \cref{Masterappr} incorporates CAM binding dynamics, that  dynamics is purely deterministic (as it is generated by ODEs), whereas the adhesion term still has a stochastic origin, as it is derived from a redistribution kernel. 
 
 Since we aim at carefully modelling the CAM binding dynamics, we need to recall a relevant  extension of the KTE framework.
 \begin{enumerate}[label=M\arabic*,ref=M\arabic*,resume=method]
\item\label{KTAPmethod}
 {The kinetic theory of active particles (KTAP) \cite{bellom3} extends \cref{KTEappr} to the settings where there are non-physical 'activity' variables that characterise the state of a cell along with its position and velocity. 
 In the context of cell adhesion, this framework allowed to incorporate  integrin binding dynamics. This  class of CAMs mediates cell-tissue rather than cell-cell interactions. 
 The corresponding models were first developed in \cite{KelkelSurulescu2011,KelkelSurulescu2012} to describe cancer invasion, for which cell-tissue interactions are a prerequisite. Viewing the proportion of bound integrins of a cell to molecules of a signal as an activity variable, the authors derived multiscale systems that couple a KTE for mesoscopic cell density and macroscopic reaction(-diffusion) equations for chemical signals. This approach was taken further in \cite{EHKS}, where a formal upscaling was performed.  Binding and unbinding of integrins was assumed to equilibrate very quickly. For other, much slower, processes, a standard parabolic scaling was adopted. The result was an equation containing myopic diffusion and local advection.
 
For free-swimming cells moving in response to a chemical signal, similar KTEs were constructed with  active variables being certain characteristics of cell internal state \cite{PST,PSTY}. Experimenting with different types of terms and scalings led to non-standard terms on the macroscale, such as, e.g. fractional diffusion \cite{PST} or
 flux-limited chemotaxis \cite{PSTY}. In these works, the upscaling was done in a rigorous way.

 So far, no non-local interactions of activity variables on the microscale have been considered. }
 \end{enumerate}
 
Inspired by \cref{SDEappr}-\cref{KTAPmethod} as well as  approaches to microscale modelling of the deterministic portion of cell motion preceding a KTE in \cite{DKSS2022,ZSMM} and to KTE upscaling in \cite{ZSMM}, we use a multiscale approach to formally derive two classes of non-local CCA models: firstly, a basic model without CAM binding dynamics in \cref{SecBasis} and, secondly, a considerably more involved model which includes such dynamics in \cref{SecCAMs}. 
Our derivations go through the following series of steps.   
\begin{enumerate}
[label=(\arabic*),ref=(\arabic*)]
\item\label{plan1} Develop a detailed microscopic description of the deterministic part of the evolution of cells and, in the case of the second model, also of their CAMs;
\item\label{plan2} lift the modelling to the mesoscopic level of a CTE using empirical measures;
\item\label{plan3} introduce a turning kernel to account for stochastic fluctuations, yielding a KTE;
\item\label{plan4} perform  the parabolic and hyperbolic  upscalings to obtain macroscopic IPDEs.                                                                                                                                                                                                                                                                                                                       \end{enumerate}
In both cases, the resulting macroscopic cell density satisfies an IPDE with myopic diffusion and non-local adhesion. In our second model \cref{MacroModel} the adhesion strength is proportional to the fraction of bound CAMs. This quantity satisfies, together with the cell density, a novel non-linear integral equation. 

Our strategy benefits from the accuracy and flexibility allowed by both the microscale modelling of the deterministic motion component and the mesoscale modelling of stochastic velocity changes. In particular, it allows to avoid direct modelling and handling of stochasticity. They are part of method \cref{SDEappr} and are often challenging.

We would like to stress that this work aims at developing a modelling framework and an  understanding of what type of CCA models we should expect on the macroscale. We do not address a specific situation that would correspond to a concrete experiment. 
Another word of caution concerns the upscaling procedures, which are only done formally. A rigorous verification, such as was carried out in \cite{ZSMM}, is beyond the scope of the present work.

\bigskip
The remainder of the paper is organised as follows. After introducing some notation in \cref{SecNot}, we first derive a basic model without CAM binding dynamics in \cref{SecBasis} and then a more involved model that includes such dynamics in  \cref{SecCAMs}. 
In \cref{SecNum}, we present and discuss the results of one-dimensional simulations for the model in \cref{SecBasis}. Finally, we discuss and summarise our findings in \cref{SecDisc}.

\section{Notation}\label{SecNot} 
In this Section we introduce some notations that are used throughout this paper. 
\begin{itemize}
\item  We denote by $|\cdot|$ the length of a vector but also the volume of a set in $\R^d$, $d\in\N$.  
\item For $\rho>0$ we set
\begin{align*}   
B_{\rho}:=&\{\theta\in\R^d: |\theta|< \rho\},\\ 
\overline{B}_{\rho}:=&\{\theta\in\R^d: |\theta|\leq \rho\},\\
S_{\rho}:=&\{\theta\in\R^d: |\theta|= \rho\}.
\end{align*}

\item Several physical variables have their traditional meaning, i.e. $t\in[0,\infty)$, $x\in\R^d$, and $v\in\R^d$ stand for time, position in space, and velocity, respectively, the space dimension being $d\in\N$. In the context where these and, in \cref{SecCAMs}, yet  another variable, $y$, serve as independent variables, we refer to $t$ and $x$ as macroscopic variables, $v$ and $y$ being referred to as non-macroscopic or mesoscopic.

\item Convolution with respect to variable $x$ is denoted by $\star$.

\item  When integrating with respect to a variable $w\in W\subset\R^k$, $k\in\N$, we use the notation $$\overline{c}^{w}:=\int_W c\,dw$$
if $c$ is a function defined on $W$ and it is evident from the context what $W$ is. Similarly, if $c$ is a measure on $W$, we set
$$\overline{c}^{w}:=\int_W c(dw)$$
to be the integral with respect to that measure. 
If $w$ is a vector consisting of {\it all} non-macroscopic variables, we omit the upper index and write $\overline{c}$ instead. We refer to a density moment that is obtained through such integration as  macroscopic moment.

\item If $z_0$ is a point in $\R^k$, $k\in\N$, then $\delta_{z_0}$ denotes the Dirac delta distribution centred at $z_0$.

\item We omit arguments of functions in many instances in order to simplify the notation.\end{itemize}

\section{Modelling without CAM binding dynamics}\label{SecBasis}
In this Section we formally derive a basic non-local diffusion-adhesion model under the assumption that adhesion originates directly from cell-cell interactions, thus ignoring the subcellular CAM binding dynamics at this stage. We follow steps \cref{plan1}-\cref{plan4} outlined in \cref{SecIntroMod}.
\subsection{Microscale model}\label{SecMicro1}
Similar to \cite{ZSMM,DKSS2022,CKSWNS2021}, we  begin with a detailed description of the deterministic part of the cell movement on the microscale which includes acceleration due to external forces.  
Let a population of a large number $$1\ll N\in\N$$ of cells be modelled as points with position and velocity coordinates
\begin{align*}
 (x_i,v_i)\in \R^d\times \R^d,\qquad i\in\{1,\dots,N\}.
\end{align*}
Following Newton's second law, we set up an initial value problem (IVP) for an ODE system that describes their  motion: 
\begin{subequations}\label{micro}
\begin{align}
&\frac{dx_i}{dt}=v_i, \\
&\frac{dv_i}{dt}=-av_i
+\chi(\cdot,x_i)\frac{1}{N}\sum_{\underset{j\neq i}{j=1}}^N\nabla_x H_r(x_i-x_j),\label{Eqvi}\\
&(x_i,v_i)(0)=(x_{i0},v_{i0}),\
\end{align}
where $(x_{i0},v_{i0})\in\R^d\times \R^d$,  $i\in\{1,\dots,N\}$, $a,r>0$,
\end{subequations}
\begin{align}
  H_r(x):=
\frac{1}{r|B_r|}\int_{\min\{|x|,r\}}^rF(s)\,ds,\label{AdhPot}
\end{align}
and
\begin{align*}
F:[0,r]\rightarrow [0,\infty),\qquad \chi:[0,\infty)\times\R^d\rightarrow [0,\infty),\ \chi=\chi(t,x),       
\end{align*}
are some continuous
functions.

 As in \cite{DKSS2022,ZSMM}, the term $(-av)$ on the right-hand side of \cref{Eqvi} is included to describe the acceleration (rather, deceleration) due to the viscous force. Following Stokes\rq\ law, we take it to be proportional to the velocity of the cell. 
 
 The reminder of the right-hand side of \cref{Eqvi} describes the acceleration due to CCA forces. This diverts from the choices made in  \cite{ZSMM,DKSS2022,CKSWNS2021} where the external forces that acted on cells were local and solely due to macroscopic signals. It is also different from \cref{SDEappr} because no simplifying assumptions are made that would allow  to reduce the ODE system \cref{micro} to a single first-order ODE for $x_i$.

 The adhesion force, a special case of the interaction force, is the sum of forces due to interaction with individual cells within reach. The scaling by $1/N$ before the sum in \cref{Eqvi} corresponds to the mean field assumption that we adopt here. Similar to \cite{Armstrong2006,Middleton,ButtenHGP2018} and many other works, we assume that interaction occurs only within a sensing region that has the form of a ball of a fixed sensing radius $r$ centred at the cell's position. The case of a more realistic sensing radius that could be a function of the physical variables, as proposed in \cite{LoyPrez20202}, is not considered here. To account for a possible spatial heterogeneity of sensitivity to adhesion, we  multiply by a parameter function $\chi$ instead. Note that in general the resulting interaction kernel
 \begin{align}
  K(t,x,x'):=\chi(t,x)\nabla_xH_r(x-x')\label{K}%
 \end{align}
is not skew-symmetric with respect to the spatial variables $x$ and $x'$ and depends on $t$. Unlike the collision of physical particles that follows Newton's third law, here we allow for non-mechanical influences on the strength of the force one cell exerts on the other. 
 
 As is standard practice, we assume the interaction force between two individual to be proportional to the  gradient of a potential, $H_r$, which we refer to as    adhesion potential.  Function $F$ describes the dependence of the adhesion strength upon the  distance relative to $r$. The chosen domain of integration in \cref{AdhPot} ensures that no interaction occurs outside the sensing region.  
 The gradient of $H_r$ computes to 
 \begin{align}
  \nabla_x H_r(x)=\begin{cases}-\frac{1}{r|B_r|}\frac{x}{|x|}F(|x|)&\text{in }B_r\backslash\{0\},\\
  0&\text{in }\R^d\backslash\overline{B}_r, 
\end{cases}
\label{nablaH}
 \end{align}
 and, unless $F(0)=F(r)=0$, it fails to exist at $0$ and/or on $S_r$. One could avoid this problem by replacing $\nabla H_r$ in \cref{Eqvi} by a function that extends it to the whole of $\R^d$. 
 We ignore the issue, assuming that cells do not accumulate on lower-dimensional sets such as points and spheres of radius $r$. Even if a small proportion of cells happens to be at a distance exactly zero or $r$ from a certain cell at some point in time, the corresponding contribution to the right-hand side of \cref{Eqvi} would then be small due to the factor $1/N$.   
 
  As seen in \cref{SecMacro1} below, the $r$-dependent coefficient $1/(r|B_r|)$ in \cref{nablaH} appears  before the non-local  adhesion term on the macroscale. We give our motivation for its inclusion as well as argue that $F(0)$ needs to be non-zero later in  \cref{SecMacro1}.

Finally, since cell speeds cannot become arbitrary large, we want to ensure that they are contained in the ball ${{B}_{s}}$ for some unattainable upper bound $s>0$. 
 A suitable rescaling turns $s$ into $1$. Thus, from now on we require
 that $v_i$'s  do not leave the velocity space
 \begin{align*}
  V:={B_1}.
 \end{align*}
Basic ODE theory guaranties this under the condition 
\begin{align}
 \frac{1}{r|B_r|}\sup_{[0,\infty)\times\R^d}\chi\sup_{[0,r]} F\leq a\label{AssumpBnd}
\end{align}
due to \cref{nablaH}.
\begin{Remark}[Well-posedness for \cref{micro}]\label{RemWPmicro1}
 System \cref{micro} can be rewritten in the form 
 \begin{align}
  &\frac{dz}{dt}=g(\cdot,z),\nonumber%
  \\
  &z(0)=z_0,\nonumber
 \end{align}
where 
\begin{align*}
 &z:=(x_1^T,v_1^T,\dots,x_N^T,v_N^T)^T:[0,\infty)\to (\R^d\times V)^N,
\end{align*}
and $$g:[0,\infty)\times(\R^d\times V)^N\to [0,\infty)\times(\R^d\times V)^N$$ is obtained by copying into a vector the right-hand sides of the  equations in \cref{micro} in the correct order. 
Let us now assume that both $F$ and $\chi$ are Lipschitz. As discussed above, $\nabla _xH_r$ may fail to exist on a lower-dimensional set in $\R^d$, implying that the classical well-posedness theory of first-order ODE systems cannot be used. Still, we verify in a coming paper that $\nabla H_r$ belongs to the class of the vector-valued functions of bounded variation (BV), i.e. it  possesses derivatives that are signed Radon measures.  Using the chain rule for BV functions \cite[Chapter 3 \S 3.10 Theorem 3.96]{Ambrosio}), one can deduce that $g$ inherits this property on $(\R^d\times V)^N$. Furthermore, it is evident that $g$ is essentially bounded and satisfies
\begin{align*}
 \nabla_{z}\cdot g(z)\equiv-dNa,
\end{align*}
so that the divergence is bounded. Therefore, the theory developed in \cite{Ambrosio2004} provides well-posedness of \cref{micro} in a certain generalised sense. We do not pursue this further here.
\end{Remark}

\subsection{Mesoscale model}\label{SecMeso1}
Our next step is to lift the microscopic model \cref{micro} to the mesoscale and extend it to a full KTE that also accounts for stochastic velocity changes. 

For each $N\in\N$, we introduce the time-dependent empirical measure
\begin{align}
 c_N(t,\cdot,\cdot):=\frac{1}{N}\sum_{i=1}^N\delta_{(x_i,v_i)(t)},\qquad t\in[0,\infty),\label{cN}
\end{align} 
where $(x_i,v_i)$ is the trajectory that the $i$th cell follows in the space-velocity space. 
This measure-valued function is an appropriate description of the mesoscopic population density rescaled so that the total mass is normalized to one. Each distribution $\delta_{(x_i,v_i)}$ models a point mass concentrated at $(x_i,v_i)$, i.e. the density of a cell at $x_i$ with velocity $v_i$. 

Let us assume for a moment that $H_r$ is sufficiently regular. In this  case, the classical ODE theory provides the well-posedness of  \cref{micro}. Moreover, the empirical measure $c_N$ corresponding to the solution of \cref{micro} satisfies in the distributional sense the  mean field PDE
\begin{align}
 &\nabla_{(t,x,v)}\cdot\left(\left(1,v,-av+\chi \nabla_x H_r\star \overline{c_N}\right)c_N\right)=0\label{CTEN}
\end{align}
and, obviously, also the initial condition
\begin{align}
 &c_N(0,\cdot,\cdot)=\frac{1}{N}\sum_{j=1}^N\delta_{(x_{i0},v_{i0})}.\nonumber%
\end{align}
For constant $\chi$, this is well-known, see, e.g. \cite{Golse}. The general case follows with \cref{LemA1} in the Appendix.
However, if $H_r$ is not regular enough, \cref{CTEN} may fail to make sense. In particular, $(\nabla_x H_r\star \bar{c})c$ is in general not well-defined if $\nabla_x H_r$ is not continuous and $c$ is a discrete, hence singular, measure. For the reason detailed in \cref{SecMicro1} we ignore this issue here. 

Since the population number is assumed to be large, we are  interested in the mean field limit as $N\to\infty$. This allows to deal with less concentrated, hence less singular, solutions to \cref{CTEN} that are functions and not discrete measures. Since \cref{CTEN} does not depend on $N$, it is reasonable to expect that this is the equation that is obtained in the limit, i.e. that $c_N$ converges to some $c$ that satisfies
\begin{align}
 &\nabla_{(t,x,v)}\cdot\left(\left(1,v,-av+\chi \nabla_x H_r\star \overline{c}\right)c\right)=0.\label{CTE}
\end{align}

The CTE \cref{CTE} provides the description of the deterministic part of cell movement driven by \cref{micro} on the mesoscopic level. To complete the modelling, we still need to include a term that accounts for stochastic perturbations. Since adhesion is particularly relevant in cancer (see \cref{IntroBB}), we include a turning operator that accounts for chaotic realignment with tissue fibers. Following \cite{ZSMM,HillenPainter}, we choose a basic turning operator 
\begin{align*}
 c\mapsto dq\overline{c}-c
\end{align*}
to illustrate our approach.  Here $q$ models the orientational distribution of tissue fibers and satisfies the following assumptions:
\begin{Assumptions}~
\begin{enumerate}
\item $q:\R^d\times V\to [0,\infty)$ and only depends on $x$ and $\frac{v}{|v|}$;
 \item $\bar{q}=\frac{1}{d}$.
\end{enumerate}
\end{Assumptions}
\noindent This kind of turning operator has been used in many models for cancer migration, see, e.g. references given in \cite{ZSMM}. The resulting mesoscopic equation is 
\begin{align}
 \nabla_{(t,x,v)}\cdot\left(\left(1,v,-av+\chi \nabla_x H_r\star \overline{c}\right)c\right)=dq\overline{c}-c.\label{meso}
\end{align}
It is a blend of a KTE and a mean field equation. This  doubly non-local IPDE accounts for both the  deterministic cell-cell and stochastic cell-tissue interactions.

 Due to  \cref{AssumpBnd,nablaH} we have that
 \begin{align}
  (-av+\chi \nabla_x H_r\star u)\cdot v\leq &-a|v|^2+|v|\sup_{[0,\infty)\times\R^d}\chi\, \underset{B_r}{\sup}|\nabla H_r|\|u\|_{L^1(\R^d)}\nonumber\\
  =&-a+\sup_{[0,\infty)\times\R^d}\chi\, \underset{B_r}{\sup}|\nabla H_r|\nonumber\\
  \leq&0\qquad \text{in } [0,\infty)\times\R^d\times S_1\qquad \text{for }\|u\|_{L^1(\R^d)}=1.
 \end{align}
Consequently,  the characteristics of the transport part of equation \cref{meso} that start in $\R^d\times V$ do not leave this set. Hence, 
\begin{align}
 c=0\qquad\text{in }[0,\infty)\times\R^d\times S_1\label{mesobc}
\end{align}
are admissible boundary conditions for equation  \cref{meso}

\begin{Remark}[Mean field limit for \cref{CTEN}]
  For constant $\chi$, the fact that \cref{CTE} is obtained from \cref{CTEN} in the mean field limit is a direct consequence of the results in \cite{Dobrusin}, provided that $H_r$ is smooth,  and  \cite{JabinWang} if it is not, at least for $a=0$. 
 \end{Remark}

\begin{Remark}[Solvability of \cref{meso}]\label{RemWPmeso}
We are not aware of results on solvability for such doubly non-local non-linear equations as  \cref{meso}. 
\end{Remark}

\subsection{Macroscale model}\label{SecMacro1}
In this Subsection we upscale \cref{meso}, \cref{mesobc} to obtain equations for the  macroscopic cell density. 
To begin with, we introduce a macroscopic rescaling of time and space and of functions depending on them: for $\ve\in(0,1]$ let 
\begin{align*}
 &\hat{t}=\ve^{\kappa}t,\qquad\kappa\in\{1,2\},\\ &\hat{x}=\ve x,\qquad \hat{r}=\ve r,\\
 &\hat{\chi}(\hat t,\hat x):=\chi(t,x),\qquad \hat F(\hat s):=F(s),\qquad \hat{q}(\hat x,v):=q(x,v),\\
 &\cep(\hat t,\hat x,v):=c(t,x,v).
\end{align*}
The values $\kappa=1$ and $\kappa=2$ correspond to the usual hyperbolic and  parabolic scalings, respectively. 
Under the proposed scaling we have
\begin{align}
 \nabla_x H_r\star \overline{c}(t,x)
 =&-\frac{1}{r|B_r|}\int_{B_r}\frac{y}{|y|}F(|y|)\overline{c}(t,x-y)\,dy\nonumber\\
 =&-\frac{\ve^{d+1}}{\hat {r}|B_{\hat{r}}|}
\int_{B_{\frac{\hat{r}}{\ve}}}\frac{y}{|y|}\hat F(\ve|y|)\overline{\cep}\left(\hat{t},\hat{x}-\ve y\right)\,dy
\nonumber\\
 =&-\frac{\ve}{\hat {r}|B_{\hat{r}}|}
\int_{B_{\hat{r}}}\frac{y}{|y|}F(\hat{y})\overline{\cep}\left(\hat{t},\hat{x}-\hat{y}\right)\,d\hat{y}\nonumber\\
=&\ve \nabla_{\hat{x}} H_{\hat{r}}\star \overline{\cep}(\hat{t},\hat{x}).\label{scAdh}
\end{align}
Rescaling \cref{meso,mesobc}, using \cref{scAdh},  and dropping the hats leads to 
\begin{subequations}
\begin{align}
 & \nabla_{(t,x,v)}\cdot\left(\left(\ve^{\kappa},\ve v,-a(v-v^{\ve}_*)\right)\cep\right)=dq\overline{\cep}-\cep,\label{transpce}
 \\
 &\cep=0\qquad\text{in }\R^d\times S_1,\\
 &v^{\ve}_*:=\ve\frac{1}{a}\chi \nabla_x H_r\star \overline{\cep}.\label{bce}
\end{align}
\end{subequations}
Set
\begin{align*}
 &c^{0}:=\underset{\ve\rightarrow0}{\lim} \, \cep,\\
  &c^{0}_1:={\underset{\ve\rightarrow0}{\lim} \,\partial_{\ve}\cep},\\
  &c_{01}^{\ve}:=c^{0}+\ve c^{0}_1.
  \end{align*}
We are interested in obtaining equations for the macroscopic  zero- and first-order approximations of $\cep$, i.e. $\overline{c^0}$ and $\overline{c_{01}^{\ve}}$.
Equation \cref{transpce} has the same form as equation (3.3) in \cite{ZSMM}. However, the term $v^{\ve}_*$ is not exactly of the form we considered in \cite{ZSMM}. Indeed, it depends on variable $t$ and lacks saturation. Still, since it vanishes at $\ve=0$, the very same formal derivation as was done in that work can be carried out in the present case.
In particular, one obtains equations
\begin{align}
(a+1)\partial_t\overline{c^{0}}=\frac{1}{2a+1}\frac{d}{d+2}\nabla_x\nabla_x^T: \left(\D[q]\overline{c^{0}}\right)-\nabla_x\cdot(\overline{c^{0}}\chi\nabla_x H_r\star\overline{c^{0}})\qquad\text{if }\kappa=2\text{ and }\E[q]\equiv0\label{PL}
\end{align}
and
\begin{align}
 &(a+1)\partial_t \overline{c_{01}^{\ve}}
 +\frac{d}{d+1}\nabla_x\cdot\left(\E[q]\overline{c_{01}^{\ve}}\right)\nonumber\\
 =& {\ve}\left(\frac{1}{2a+1}\frac{d}{d+2}\nabla_x\nabla_x^T:\left(\D[q]\overline{c_{01}^{\ve}}\right)-\frac{1}{(a+1)^2}\frac{d^2}{(d+1)^2}\nabla_x\cdot\left(\E[q]\nabla_x\cdot\left(\overline{c_{01}^{\ve}}\E[q]\right)\right)\right)\nonumber\\
 &-\ve \nabla_x\cdot\left(\overline{c_{01}^{\ve}}\chi\nabla_x H_r\star\overline{c_{01}^{\ve}}\right)\nonumber\\
 &+O\left(\ve^2\right)\qquad\text{if }\kappa=1,\label{HLCorr}
\end{align}
where
\begin{align*}
 &\E[q]:=\int_{S_1}\theta q(\theta )\,d\theta,\\
 &\D[q]:=\int_{S_1}\theta  \theta^T q(\theta)\,d\theta.
\end{align*}
Recalling the adhesion operator
\begin{align}
 \A_ru(x)=&\frac{1}{r|B_r|} \int_{B_r}u(x+\xi)\frac{\xi}{|\xi|}F(|\xi|)\,d\xi\label{AdhOper}\nonumber\\
 \end{align}
 from \cite{Armstrong2006} and noticing that 
 \begin{align}
  \A_ru\equiv & \nabla_x H_r\star u,\label{Arconv}
 \end{align}
 we can alternatively rewrite \cref{PL,HLCorr} as follows:
 \begin{align}
(a+1)\partial_t\overline{c^{0}}=\frac{1}{2a+1}\frac{d}{d+2}\nabla_x\nabla_x^T: \left(\D[q]\overline{c^{0}}\right)-\nabla_x\cdot(\overline{c^{0}}\chi\A_r\overline{c^{0}})\qquad\text{if }\kappa=2\text{ and }\E[q]\equiv0\label{PL_}
\end{align}
and
\begin{align}
 &(a+1)\partial_t \overline{c_{01}^{\ve}}
 +\frac{d}{d+1}\nabla_x\cdot\left(\E[q]\overline{c_{01}^{\ve}}\right)\nonumber\\
 =& {\ve}\left(\frac{1}{2a+1}\frac{d}{d+2}\nabla_x\nabla_x^T:\left(\D[q]\overline{c_{01}^{\ve}}\right)-\frac{1}{(a+1)^2}\frac{d^2}{(d+1)^2}\nabla_x\cdot\left(\E[q]\nabla_x\cdot\left(\overline{c_{01}^{\ve}}\E[q]\right)\right)\right)\nonumber\\
 &-\ve \nabla_x\cdot\left(\overline{c_{01}^{\ve}}\chi\A_r\overline{c_{01}^{\ve}}\right)\nonumber\\
 &+O\left(\ve^2\right)\qquad\text{if }\kappa=1.\label{HLCorr_}
\end{align}

 Both \cref{PL_,HLCorr_}  contain the same $q$-dependent terms, such as, e.g. the myopic diffusion
 \begin{align*}
  \nabla_x\nabla_x^T:\left(\D[q]u\right),
 \end{align*}
 as the corresponding equations (3.18) and (3.51) from \cite{ZSMM}. We refer to that work for a discussion of these terms, as well as for the  formulas for the mesoscopic approximations $c^0$ and $c_1^0$.  In contrast to \cite{ZSMM}, our new equations \cref{PL_,HLCorr_} also contain the non-local advection term
\begin{align*}
 -\nabla_x\cdot(u\chi\A_ru)
\end{align*}
of the form originally proposed in \cite{Armstrong2006} to model CCA. That  model is a special case of \cref{PL_} and corresponds to $\chi$ and $q$ being constant. 

Another way to rewrite, e.g. \cref{PL_} is by decomposing the spatial motion into an anisotropic diffusion in the divergence form and advection:
 \begin{align}
(a+1)\partial_t\overline{c^{0}}=\frac{1}{2a+1}\frac{d}{d+2}\nabla_x\cdot\left(\D[q]\nabla_x\overline{c^{0}}\right)+\nabla_x\cdot \left(\overline{c^{0}}\left(\frac{1}{2a+1}\frac{d}{d+2}\nabla_x\cdot \D[q]-\chi\A_r\overline{c^{0}}\right)\right).
\label{PL_2}
\end{align}
In \cref{SecNum} below we show results of some numerical simulations for this model.

 The heuristic analysis in \cite{GerChap2008} as well as  the rigorous study in \cite{EPSZ} showed that in the limit as $r\to0$ the non-local adhesion operator ${\cal A}_r$ approaches the (local) spatial gradient, provided that $F(0)=d+1$. This is the expected limit behaviour and the reason for the factor $1/(r|B_r|)$ in \cref{AdhOper} and the observation  that $F(0)\neq0$, see  \cref{SecMicro1}. 

\begin{Remark}[Non-local operator $\mathring{\nabla}_{r}$] 
 Since our approach to the derivation of the mesoscopic equation \cref{meso} from the microscopic ODE system \cref{micro} is based on empirical measures, it limits the admissible choices of the interaction potential. 
 If we were to start directly on the mesoscopic level and would only accept integrable densities $c$ rather than discrete measures, then we could choose a discontinuous potential such as
 \begin{align}
  H_r(x):=\begin{cases}
           1&\text{in }B_r,\\
  0&\text{in }\R^d\backslash\overline{B}_r. 
\end{cases}
 \end{align}
 In this case, the gradient of $H_r$ is a measure. In higher dimensions $d\geq2$ it is given by
 \begin{align*}
  \nabla_x H_r=-id_x\,dS_r, 
 \end{align*}
 which leads to
 \begin{align}
  \nabla_x H_r\star u=&\int_{S_r}u(x+\xi)\xi\,d\xi\nonumber\\
  =&\mathring{\nabla}_{r}u.
 \end{align}
The latter is the non-local operator that was previously introduced in \cite{othmer-hillen2} to describe non-local chemotaxis. 
\end{Remark}
  
 \begin{Remark}[Solvability of \cref{PL_,HLCorr_}]
  Several works established solvability of non-local diffusion-adhesion equations as well as systems containing them, see \cite{reviewNonlocal2020} and references therein. Yet none of them included the case of a myopic diffusion.  Solvability in the presence of a scalar \cite{WinSur2017} or a tensor \cite{Heihoff} myopic diffusion  has so far been accomplished for models with advection due to haptotaxis, i.e. a directed movement along the spatial gradient of the macroscopic tissue density, rather than adhesion.
  
 \end{Remark}
 \begin{Remark}[Rigorous upscaling of \cref{meso}]\label{Remupscale}
  The presented meso-to-macro upscaling  is formal. 
  Unlike the case that was handled in \cite{ZSMM}, \cref{meso} is a non-linear equation and, as previously observed in \cref{SecMeso1}, the term $u\A_ru$ is not defined for singular measures $u$. Thus, the approach that we developed in \cite{ZSMM} is not directly applicable. Indeed, there we relied on the linearity of the KTE and on the possibility of considering measure-valued solutions. A rigorous upscaling for \cref{meso} remains an open question. 
 \end{Remark}

\section{Modelling with CAM binding dynamics} \label{SecCAMs}
In this Section we derive a new non-local diffusion-adhesion model that takes into account subcellular CAM binding dynamics. We adhere to steps \cref{plan1}-\cref{plan4} outlined in \cref{SecIntroMod}.

\subsection{Microscale model}
The basic microscopic model \cref{micro}  neglects the CAMs binding dynamics, which, as pointed out in the Introduction, is the underlying mechanism of cell-cell binding. In \cite{ButtenHGP2018}, this mechanism  was taken into account. There it was assumed that at each time and position in space a single cell is moving while the rest of the population in its background is effectively standing still. The interactions between the cell and the background population were described by reversible reactions that either discriminate between bound/free CAMs or not. 

In this Subsection,  we exploit the microscopic approach that allows to describe mutual  interactions between the CAMs of a pair of cells. For this purpose, we construct a system of ODEs that includes equations not just for $x_i$ and $v_i$ but also for $y_i$, the proportion of bound CAMs of $i$th cell. This bears resemblance to modelling in  \cite{KelkelSurulescu2011,KelkelSurulescu2012}, although, in our case, the interactions occur inside the population  rather than with an external signal.
 
For the reader's convenience, we  first collect  all involved model  variables and parameters, including those  previously introduced in \cref{SecBasis}:
\begin{itemize}
 \item $1\ll N\in\N$: population number;
 \item $t\in[0,\infty)$: time;
 \item $d\in\N$: space dimension;
 \item  $x_i:[0,\infty)\rightarrow\R^d$, $x_i=x_i(t)$: position of $i$th cell in space;
 \item $V=B_1$: velocity space;
 \item $v_i:[0,\infty)\rightarrow V$, $v_i=v_i(t)$: velocity of $i$th cell;
 \item $a\in(0,\infty)$: cell deceleration rate; 
\item $r\in(0,\infty)$: cell sensing radius;
 \item $R\in (0,\infty)$: total number of CAMs  of a cell;
 \item $y_i:[0,\infty)\rightarrow(0,1)$, $y_i=y_i(t)$: proportion of bound CAMs of  $i$th cell; 
 \item $S:[0,\infty)\times\R^d\rightarrow[0,\infty)$, $S=S(t,x)$: concentration 
of a chemical on which the likeliness to bind/unbind depends;
\item $k^{+}/k^{-}:[0,\infty)\times[0,\infty)\rightarrow [0,\infty)$
: CAM  binding/unbinding rate constants, depend on $S$ and the distance
between interacting cells;
 \item $F:[0,r]\rightarrow[0,\infty)$: distance-dependent component of adhesion force;
 \item $H_r$: adhesion potential as defined by \cref{AdhPot};
 \item $\chi:[0,\infty)\times\R^d\rightarrow[0,\infty)$, $\chi=\chi(t,x)$: adhesion sensitivity.
\end{itemize}
We begin by describing the  cell motion: for  $i\in\{1,\dots,N\}$
\begin{subequations}\label{micro_new}
\begin{align}
\frac{dx_i}{dt}=&v_i, \\
\frac{dv_i}{dt}=&-av_i
+\chi(\cdot,x_i)y_i\frac{1}{N}\sum_{\underset{j\neq i}{j=1}}^N\nabla_x H_r(x_i-x_j).\label{Eqvi_new}
\end{align}
\end{subequations}
 Unlike  \cref{micro}, in \cref{micro_new} the total adhesion force acting on the $i$th cell is taken to be proportional to $y_i$. This implies that the more bonds a cell has, the stronger this force is. 

We make the following assumptions on the CAM binding: 
 \begin{enumerate}[label=(A\roman*),ref=(A\roman*)]
 \item\label{Ass1}a single type of CAMs influences cell motion;
 \item all cells have exactly the same number of these CAMs;
 \item each cell has  both bound and free CAMs at every point in time;
 \item a CAM of a cell can only bind to a distinct CAM of another cell;
  \item  CAMs of a pair of cells can bind only if the distance between the cells  %
  is smaller than the sensing radius $r$; 
  \item if the distance
approaches $r$, then their CAMs unbind;
\item  binding/unbinding kinetics  obeys an analog of the mass action law;
\item\label{Assl}  binding is reversible; the corresponding rates for a pair of cells depends on the distance between them and the concentration of a  chemical at the middle distance.
 \end{enumerate}
An example of an application that we have in mind here is a simplified description of  the formation of adherens junctions through the calcium-mediated cadherin binding, see \cref{SecIntroMod} above.

  Let ${\cal F}_i$ and ${\cal B}_i$ denote respectively the free and bound CAMs of the $i$th cell. 
Then, the above assumptions can be described by the  following \lq reactions\rq: for  all $i,j\in\{1,\dots,N\}$, $i\neq j$,
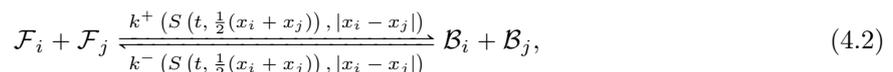
 \begin{align}
  &\ce{{\cal F}_i + {\cal F}_j <=>[$k^{+}\left(S\left(t,\frac{1}{2}(x_i+x_j)\right),|x_i-x_j|\right)$][$k^{-}\left(S\left(t,\frac{1}{2}(x_i+x_j)\right),|x_i-x_j|\right)$]  {\cal B}_i + {\cal B}_j},
 \label{bunb1}
 \end{align}
where
\begin{align}
 k^{\pm}\equiv 0\qquad \text{in }[0,\infty)\times [r,\infty).\label{kpm0}
\end{align}
It is reasonable to suppose  in this context that the following quantities replace the standard chemical concentrations:
 \begin{align*}
  [{\cal B }_i]:=&\frac{[\text{number of bound CAMs of }i\text{th cell}]}{ [\text{total number of CAMs in population}]\cdot[\text{volume of the sensing region}]}\nonumber\\
  =&\frac{y_iR}{NR|B_r|}\nonumber\\
  =&\frac{y_i}{N|B_r|}
 \end{align*}
 and, similarly,
 \begin{align*}
  [{\cal F }_i]:=\frac{1-y_i}{N|B_r|}.
 \end{align*}
 
 We provide an example of binding/unbinding rates $k^{+}$/$k^{-}$.
 \begin{Example}
For some non-decreasing  $K^+:[0,\infty)\to [0,\infty)$ and non-increasing $K^-:[0,\infty)\to [0,\infty)$ and constants $a^{\pm},b^{\pm}\in(0,\infty)$ set
 \begin{align*}
 &k^{\pm}(S,\rho): =K^{\pm}(S)\varphi^{\pm}(\rho)\qquad\text{for }S,\rho\in[0,\infty),\\
 &\varphi^{\pm}(\rho):=
 \begin{cases}
\left(r^{b^{\pm}}-\rho^{b^{\pm}}\right)^{\pm a^{\pm}}&\text{for }\rho\in[0,r),\\
0&\text{for }\rho\in[r,\infty).
\end{cases}
 \end{align*}
This choice produces a speedy detachment when the cell distance is close to $r$. 
\end{Example}
In order to estimate the  binding/unbinding rates in  a concrete type  of cell-cell binding, the general framework developed in \cite{Bell} can be adopted. There, the binding process is decomposed into two phases: the formation of the encounter complex and the actual bond formation, both being reversible reactions. The reaction rates for the former are described by functions of the cell separating distance and the  translational
diffusion coefficient for CAM motion
in the cell membrane.
Specifically in the case of cadherin binding, higher calcium concentrations correlate with faster diffusion of cadherins in the cell membrane \cite{Leckbandetal2016}. As to the bond formation, it is well-understood \cite[Chapter 19]{Albertsetal} that  calcium is indispensable for   cadherins to achieve the rigid structure that is necessary for them to bind. These are the observations that have led us to assume that $k^{\pm}$ are functions of distance and concentration of a chemical that mediates the binding of CAMs.

 Applying the law of mass action to \cref{bunb1},  we arrive at the following ODE system for the dynamics of $y_i$'s:  
 \begin{subequations}\label{syst_yi}
 \begin{align}
  \frac{d y_i}{dt}=&
  \frac{1}{N}\sum_{\underset{j\neq i}{j=1}}^NG_r[S](t,(x_i,y_i),(x_j,y_j)),\qquad i\in\{1,\dots,N\},
  \label{Eq_yi}
 \end{align}
 where 
 \begin{align}
  &G_r[S](t,(x,y),(x',y')):=G_r^{+}[S](t,x,x')(1-y)(1-y')-G_r^{-}[S](t,x,x')yy',\label{G}\\
  &G_r^{{\pm}}[S](t,x,x'):=\frac{1}{|B_r|}k^{{\pm}}\left(S\left(t,\frac{1}{2}(x+x')\right),|x-x'|\right).\label{Gpm}
 \end{align}
 \end{subequations}
  System  \cref{micro_new}, \cref{syst_yi} is our new microscopic model.

 Since  $k^{\pm}\geq0$, we have that
\begin{align}
 G_r[S](t,(x,0),(x',y'))\geq0,\ G_r[S](t,(x,1),(x',y'))\leq 0\qquad\text{for all }t\in[0,\infty),\ x,x'\in\R^d,\ y'\in[0,1].\label{pos}
\end{align}
Standard ODE theory implies that  $y_i$'s do not leave $[0,1]$.

\begin{Remark}[Well-posedness of \cref{micro_new}, \cref{syst_yi}]
 The well-posedness of \cref{micro_new}, \cref{syst_yi} can be addressed in the same way as for \cref{micro}, see \cref{RemWPmicro1}.  
Depending on the choice of functions $k^{\pm}$ one could avail of the classical ODE existence and uniqueness results or the more general theory from \cite{Ambrosio2004}. In particular, since the  divergence of the right-hand side of the ODE system for $(x_1^T,v_1^T,y_1,\dots,x_N^T,v_N^T,y_N)^T$ computes to 
 \begin{align}
  &-dNa-\frac{1}{N}\sum_{i=1}^N\sum_{\underset{j\neq i}{j=1}}^N(G_r^{+}[S](t,x_i,x_j)(1-y_j)+G_r^{-}[S](t,x_i,x_j)y_j)\nonumber\\
  =&-dNa-\frac{1}{N}\sum_{j=1}^N\left((1-y_j)\sum_{\underset{j\neq i}{i=1}}^NG_r^{+}[S](t,x_i,x_j)+y_j\sum_{\underset{j\neq i}{i=1}}^NG_r^{-}[S](t,x_i,x_j)\right),
 \end{align}
 the essential boundedness of $k^{\pm}$ is necessary for the results from \cite{Ambrosio2004} to apply.
\end{Remark}

 \subsection{Mesoscale model}
 In this Subsection, we lift the microscopic model \cref{micro_new}, \cref{syst_yi} to the mesoscale and extend it to a full KTE that  includes stochastic velocity changes. 
 Similar to \cref{SecMeso1}, we begin by introducing the empirical measures
\begin{align}
 c_N(t,\cdot,\cdot,\cdot):=\frac{1}{N}\sum_{j=1}^N\delta_{(x_i,v_i,y_i)(t)}.\nonumber
\end{align} 
 A formal application of  \cref{LemA1} to system \cref{micro_new}, \cref{syst_yi} and $Z:=(x^T,v^T,y)^T$
  yields that $c_N$ solves in the distributional sense the mean field IVP
  \begin{subequations}
\begin{align}
&\nabla_{(t,x,v,y)}\cdot\left(\left(1,v,-av+\chi y\nabla_x H_r\star \overline{c_N},\G_r[S]\overline{c_N}^{v}-\frac{1}{N}\tr G_r[S]\right)c_N\right)=0,\label{CTEN2}\\ 
&c_N(0,\cdot,\cdot,\cdot)=\frac{1}{N}\sum_{j=1}^N\delta_{(x_i,v_i,y_i)(0)}.\nonumber
\end{align}
\end{subequations}
where
\begin{align*}
 &\G_r[S]u\,(t,x,y):=\int_{0}^1\int_{\R^d}G_r[S](t,(x,y),(x',y'))u(t,x',y')\,dx'dy',\\
 &\tr G_r[S](t,x,y):=G_r[S](t,(x,y),(x,y)).
\end{align*}
Here, as in \cref{SecMeso1}, we ignore the potential discontinuities in the kernels.
Passing formally to the limit as $N\rightarrow\infty$ in \cref{CTEN2}, we arrive at the mean field limit equation
\begin{align}
\nabla_{(t,x,v,y)}\cdot\left(\left(1,v,-av+\chi y\nabla_x H_r\star \overline{c},\G_r[S]\overline{c}^{v}\right)c\right)=0.\label{CTE_new}
\end{align}
To account for chaotic interactions with tissue we use the same turning operator as in \cref{meso}.
The resulting mesoscopic equation is thus:
\begin{align}
 \nabla_{(t,x,v,y)}\cdot\left(\left(1,v,-av+\chi y\nabla_x H_r\star \overline{c},\G_r[S]\overline{c}^{v}\right)c\right)=dq\overline{c}^{v}-c.\label{meso_new}
\end{align} 

Due to  \cref{AssumpBnd,nablaH} we have that
 \begin{align}
  (-av+\chi y\nabla_x H_r\star u)\cdot v\leq &-a|v|^2+|v|\underset{B_r}{\sup}|\nabla H_r|\|u\|_{L^1(\R^d)}\nonumber\\
  =&-a+\underset{B_r}{\sup}|\nabla H_r|\nonumber\\
  \leq&0\qquad \text{in } [0,\infty)\times\R^d\times S_1\times[0,1]\qquad \text{for }\|u\|_{L^1(\R^d)}=1.\label{AssumpBnd_}
 \end{align}
Further, \cref{pos} implies that 
\begin{subequations}\label{posGop}
\begin{align}
&\G_r[S]u\geq0 \qquad\text{in }[0,\infty)\times\R^d\times\{0\},\qquad\text{ for } u\geq0,\\
&\G_r[S]u\leq0 \qquad\text{in }[0,\infty)\times\R^d\times\{1\}, \qquad\text{ for } u\geq0.
\end{align}
\end{subequations}
 Combining \cref{AssumpBnd_,posGop}, we conclude that the characteristics of the transport part of equation \cref{CTE_new} that start in  $\R^d\times V\times[0,1]$
do not leave this set. Hence, 
\begin{align}
 c=0\qquad\text{in }[0,\infty)\times\R^d\times ((\overline{B_1}\times\{0,1\})\cup (S_1\times(0,1)))\label{bc_new}
\end{align}
 are admissible boundary conditions for \cref{meso_new}.
\begin{Remark}[Rigorous treatment of \cref{CTEN2}]
 We are not aware of rigorous results on well-posedness or mean field limit for such CTEs as \cref{CTEN2}. Note that unlike standard applications arising in physics, the kernels of the integral operators we are dealing with here are not skew symmetric. 
\end{Remark}
\begin{Remark}[Solvability of \cref{meso_new}] Equation \cref{meso_new} is a generalisation of \cref{meso}. As mentioned in \cref{RemWPmeso}, the solvability of the latter  equation has not been addressed so far.
\end{Remark}
\subsection{Macroscale model}
In this Subsection we upscale \cref{meso_new}, \cref{bc_new} to obtain equations for the  macroscopic cell density. 
As in \cref{SecMeso1}, we begin by introducing a macroscopic rescaling of time and space and of functions depending on them: for $\ve\in(0,1]$ let 
\begin{align*}
 &\hat{t}=\ve^{\kappa}t,\qquad\text{for }\kappa\in\{1,2\},\\ 
 &\hat{x}=\ve x,\qquad \hat{r}=\ve r,\qquad \hat{d}=\ve d,\\
 &\hat{S}(\hat t,\hat x):=S(t,x),\qquad\hat{\chi}(\hat t,\hat x):=\chi(t,x),\qquad \hat F(\hat s):=F(s),\qquad \hat{q}(\hat x,v):=q(x,v),\\
 &\hat k^\pm(S,\hat d):=\ve^{-\mu}k^\pm(S,d)\qquad \text{for some }\mu>0,\\
 &\cep(\hat t,\hat x,v,y):=c(t,x,v,y).
\end{align*}
As before, we consider here  hyperbolic ($\kappa=1$) and parabolic ($\kappa=2$) space-time scalings. The rescaling of reaction rate constants means rescaling of $dy/dt$. A negative epsilon power is chosen to reflect the fact that the CAM binding and unbinding are the fastest among  all included processes. 
Rescaling \cref{meso_new,bc_new}  and dropping the hats leads to 
\begin{subequations}
\begin{align}
 &\nabla_{(t,x,v,y)}\cdot\left(\left(\ve^{\kappa},\ve v,-av+\ve\chi y\nabla_x H_r\star \overline{\cep},\ve^{-\mu}\G_r[S]\overline{\cep}^{v}\right)\cep\right)=dq\overline{\cep}^{v}-\cep,\label{transpceCAMs}\\
 &\cep=0\qquad\text{in }{[0,\infty)\times\R^d\times ((\overline{B_1}\times\{0,1\})\cup (S_1\times(0,1)))}.\label{bc_new_ve}
\end{align}
\end{subequations}

Following the approach from \cite{ZSMM}, we derive equations connecting some zero, first, and second moments of $\cep$. 
To begin with,  we integrate \cref{transpceCAMs} by parts with respect to $(v,y)$ over $V\times[0,1]$ using \cref{bc_new_ve} and then divide by $\ve^{\kappa}$ in order to obtain an equation which connects the macroscopic zero and first order $v$ moments:
\begin{align}
 \partial_t \overline{\cep}+\ve^{1-\kappa}\nabla_x\cdot \overline{v\cep}=0.\label{mom0}
\end{align}
Next, we multiply \cref{transpceCAMs} by $v$ and  once again integrate by parts over $V\times[0,1]$ using \cref{bc_new_ve}:
\begin{align}
\ve^{\kappa}\partial_t \overline{v\cep}+\ve\nabla_x\cdot \overline{vv^T\cep}+a\overline{v\cep}-\ve\chi \overline{y\cep}\nabla_x H_r\star \overline{\cep}=
 \frac{n}{n+1}\E[q]\overline{\cep}-\overline{v\cep}.\label{M1e}
\end{align}
Rearranging and dividing \cref{M1e} by $\ve^{\kappa-1}$ leads to 
\begin{align}
 -(a+1)\ve^{1-\kappa}\overline{v\cep}=&\ve^{2-\kappa} \nabla_x\cdot \overline{vv^T\cep}-\ve^{2-\kappa}\left(\chi \overline{y\cep}\nabla_x H_r\star \overline{\cep}+\ve^{-1}\frac{n}{n+1}\E[q]\overline{\cep}\right)+\ve\partial_t \overline{v\cep}.\label{sc2}
\end{align}
Next, we apply $(\nabla_x\cdot)$ to both sides of \cref{sc2} and plug the expression on the right-hand side into \cref{mom0}. In order to eliminate the resulting term with the mixed derivative $(\nabla_x\cdot)\partial_t$ we apply $\varepsilon^{\kappa}\partial_t$ to both sides of \cref{mom0}. Thus we arrive at the following differential equation for the macroscopic  moments of zero and second order:
\begin{align}
 \ve^{\kappa}\partial_{t^2} \overline{\cep}+(a+1)\partial_t \overline{\cep}
 =&\ve^{2-\kappa}\nabla_x\nabla_x^T:\overline{vv^T \cep}-\ve^{2-\kappa}\nabla_x\cdot\left(\chi \overline{y\cep}\nabla_x H_r\star \overline{\cep}+\ve^{-1}\frac{n}{n+1}\E[q]\overline{\cep}\right).\label{mom012}
\end{align}

Now we are ready to start the limit procedure. Let $$c^{0}:=\underset{\ve\rightarrow0}{\lim} \, \cep.$$
Sending $\ve$ to zero in \cref{bc_new_ve}, we obtain
\begin{align}
 c^0=0\qquad\text{in }[0,\infty)\times\R^d\times ((\overline{B_1}\times\{0,1\})\cup (S_1\times(0,1))).\label{bc0}
\end{align}
Multiplying  \cref{transpceCAMs} by $\ve^{\mu}$ and passing to the limit as $\ve\to0$ yields an equation for $c^0$: 
\begin{align}
 \partial_y\left(\G_r[S]\overline{c^0}^{v}c^0\right)=0.\label{eqy}
\end{align}
Integrating \cref{eqy} using \cref{bc0}   yields
\begin{align}
 \G_r[S]\overline{c^0}^{v}c^0\equiv0.\label{eqGr}
\end{align} 
To resolve \cref{eqGr} with respect to $c^0$, we 
first need to study equation 
\begin{align}
\G_r[S]u=0   \label{EqGu}                 \end{align}
 for a given function $u$. Combining \cref{kpm0} and \cref{G}-\cref{Gpm}, $\G_r[S]u$ can be expressed in terms of moments:
\begin{subequations}
\begin{align}
 \G_r[S]u(\cdot,\cdot,y')
 =&\G_r^{+}[S](\overline u^y-\overline{yu}^y)-y'\left(\G_r^{+}[S](\overline u^y-\overline{yu}^y)+\G_r^{-}[S]\overline{yu}^y\right),\label{CG}\\ 
 \G_r^{{\pm}}[S]u(t,x):=&\int_{\R^d}G_{{r}}^{{\pm}}[S](t,x,x')u(t,x')\,dx'.
\label{Ga}
\end{align}
\end{subequations}
Using \cref{CG}, we can resolve \cref{EqGu} with respect to   variable $y$ and obtain
\begin{subequations}
\begin{align}
 &\G_r[S]u(\cdot,\cdot,y_*)=0\qquad\Leftrightarrow\qquad y_*=\Y_r[S]\left(\overline u^y,\overline{yu}^y\right)
,\label{ystar}\\
 &\Y_r[S](\mu_0,\mu_1):=\frac{\G_r^{+}[S](\mu_0-\mu_1)}{\G_r^{+}[S](\mu_0-\mu_1)+\G_r^{-}[S]\mu_1},
 \label{Ystar}
\end{align}
\end{subequations}
provided that the denominator of the fraction on the right-hand side of \cref{Ystar} is non-zero. 
Given that $k^{\pm}\geq0$, we have that 
\begin{align}
 \G_r^{+}[S](\mu_0-\mu_1),\G_r^{-}[S]\mu_1\geq0\qquad\text{for }\mu_0\geq\mu_1\geq0,\label{ineqmu00}
\end{align}
and, moreover, if 
\begin{align}
 \G_r^{+}[S](\mu_0-\mu_1)>0\qquad\text{or}\qquad\G_r^{-}[S]\mu_1>0\label{ResolvCond}
\end{align}
as well, then $\Y_r[S](\mu_0,\mu_1)$ is well-defined and 
\begin{align}
 \Y_r[S](\mu_0,\mu_1)\in[0,1].\nonumber
\end{align}
Since $y\in[0,1]$, inequalities in \cref{ineqmu00} are satisfied for 
\begin{align}
 \mu_0:=\overline u^y,\qquad \mu_1:=\overline{yu}^y.\nonumber
\end{align}
Assuming that for $u:=\overline{c^0}^{v}$ condition \cref{ResolvCond} is satisfied, i.e. that 
\begin{align}
 \mu_0:=\overline{c^0},\qquad\mu_1:=\overline{yc^0}\nonumber
\end{align}
satisfy \cref{ResolvCond}, we can resolve \cref{eqGr} with respect to the $y$-variable and obtain
\begin{align}
  c^0(\cdot,\cdot,\cdot,dy)=\overline{c^0}^y\delta_{\Y_r[S]\left(\overline{c^0},\overline{yc^0}\right)}(y).\label{c0delta}
\end{align}  
Multiplying \cref{c0delta} by $y$ and integrating over $V\times[0,1]$, we arrive at an equation that connects the macroscopic zero and first order $y$ moments of $c^0$:
\begin{align}
 &\overline{yc^0}=\overline{c^0}\Y_r[S]\left(\overline{c^0},\overline{yc^0}\right).\label{my1}
\end{align}

Next, we integrate \cref{transpceCAMs} with respect to $y$ over $[0,1]$ using \cref{bc_new_ve} and pass to the limit as $\ve\rightarrow0$ to obtain
\begin{align}
 -a\nabla_{v}\cdot\left(v\overline{c^0}^y\right)=dq\overline{c^0}-\overline{c^0}^y.\label{eqmu}
\end{align}
This equation can be resolved with  respect to $\overline{c^0}^y$ using the method of characteristics (a similar case was treated in \cite{ZSMM}). The solution reads 
\begin{align}
 \overline{c^0}^y=\overline{c^{0}}q\Cr{xi1},\label{solc0}
 \end{align}
 where
\begin{align}
 \Cl[xi]{xi1}(v)=\begin{cases}
 \frac{d}{da-1}\left(|v|^{-d+\frac{1}{a}}-1\right)&\text{for }a\neq \frac{1}{d},\\
 -\frac{d}{a}\ln|v|&\text{for }a=\frac{1}{d}.\end{cases}\label{xi1l}
\end{align}
Altogether, combining \cref{c0delta,solc0}, we obtain a formula  for $c^0$ in terms of some of its macroscopic moments:
\begin{align}
 c^0(\cdot,\cdot,\cdot,y)=\overline{c^{0}}q\Cr{xi1}\delta_{\Y_r[S]\left(\overline{c^0},\overline{yc^0}\right)}(y).
\end{align}

Next, we multiply \cref{eqmu} by $v$ and $vv^T$, respectively, and integrate by parts over $V$ using \cref{bc0} to obtain formulas for  macroscopic moments of order one,
\begin{align}
 &a\overline{vc^{0}}=\frac{d}{d+1}\E[q]\overline{c^{0}}-\overline{vc^{0}}\nonumber\\
 \Leftrightarrow\qquad & \overline{vc^{0}}=\frac{1}{a+1}\frac{d}{d+1}\E[q]\overline{c^{0}},\label{M1}
\end{align}
and two,
\begin{align}
 &2a\overline{vv^Tc^{0}}=d\overline{vv^Tq}\,\overline{c^{0}}-\overline{vv^Tc^{0}}\nonumber\\
 \Leftrightarrow\qquad &
 \overline{vv^Tc^{0}}=\frac{1}{2a+1}\frac{d}{d+2}\D[q]\overline{c^{0}}.\label{M2}
\end{align}
Passing formally to the limit in \cref{mom0} and using \cref{M1} we arrive at the CTE 
\begin{align}
 (a+1)\partial_t \overline{c^{0}}+\frac{d}{d+1}\nabla_x\cdot\left(\overline{c^{0}}\E[q]\right)=0\qquad\text{if }\kappa=1.\label{HL}
\end{align}
This is typical for hyperbolic scaling. 
Now we address the case of parabolic scaling. Passing formally to the limit in \cref{mom012}, using \cref{M2}, and recalling \cref{my1}, we arrive at 
\begin{subequations}\label{MacroModel}
\begin{align}
&(a+1)\partial_t\overline{c^{0}}=\frac{1}{2a+1}\frac{d}{d+2}\nabla_x\nabla_x^T: \left(\D[q]\overline{c^{0}}\right)-\nabla_x\cdot\left( \chi \overline{yc^{0}}\nabla_x H_r\star \overline{c^{0}}\right),\\
&\overline{yc^0}=\overline{c^0}\Y_r[S]\left(\overline{c^0},\overline{yc^0}\right)\label{momy01}\\
&\text{if }\kappa=2\text{ and }\E[q]\equiv0.\nonumber
\end{align}
\end{subequations}

Similar to \cref{PL_}, equation \cref{MacroModel} for the macroscopic cell density, $\overline{c^0}$, includes myopic diffusion and non-local adhesion.  This time, however, the sensitivity to the adhesion force acting on the cells is proportional to the amount of bounded CAMs, $\overline{yc^0}$, and, thus, to that of the adhesion bonds formed by the cells.  

Equation \cref{momy01} appears to be of a new type. In general, it is non-local and non-linear  and cannot be explicitly resolved for $\overline{yc^0}$. 
Our next Example deals with a special case where the equation is local and easy to solve.
\begin{Example}
 Let us assume that $k^{\pm}$'s  are singular measures rather than functions and of the form 
 \begin{align}
  k^{\pm}(S,\rho):=K^{\pm}(S)\delta_{0}(\rho).\label{kpmlocal}
 \end{align}
 This choice  would correspond to the impossible situation where binding would only occur locally. 
For $k^{\pm}$ from  \cref{kpmlocal}, operators $\G_r^{\pm}[S]$, as defined by \cref{Ga,Gpm}, turn into local multiplication operators:
 \begin{align*}
  \G_r^{\pm}[S]u=K^{\pm}(S)u,
 \end{align*}
and equation \cref{momy01} can be easily resolved:
\begin{align}
 \overline{yc^0}=\frac{(K^{+})^{\frac{1}{2}}}{(K^-)^{\frac{1}{2}}+(K^+)^{\frac{1}{2}}}(S)\overline{c^0}.
\end{align}
\end{Example}

\begin{Remark}[Solvability of \cref{MacroModel}]\label{RemSolvMM}
 System \cref{MacroModel} is  strongly coupled. The presence of the non-linear non-local equation \cref{momy01} of a new type as well as of a potentially degenerate myopic diffusion makes its analysis challenging. We establish a result on local well-posedness for a variant of this system which includes a quasilinear degenerate diffusion in divergence form rather than a myopic one in a coming paper. We expect that our analysis there will indicate a reliable  approach to a numerical treatment of \cref{MacroModel}.
 \end{Remark}
\begin{Remark}[Rigorous upscaling of \cref{meso_new}]
 As in the case of  \cref{meso}, the upscaling provided above for \cref{meso_new} is just formal, see \cref{Remupscale} for the comparison with the case we studied earlier in \cite{ZSMM}.
\end{Remark}

\section{Simulations}
\label{SecNum}
In this Section, we present some simulation results for  equation \cref{PL_2} and its non-myopic modification on a one-dimensional interval. Several options have been considered in the literature  regarding appropriate boundary conditions and treatment of the integrand of ${\cal A}_r$ in the area where the argument falls outside a bounded spatial domain \cite{ButtenHillen}. Here we assume no-flux boundary conditions and extend the integrand in ${\cal A}_r$ by zero. Denoting $u:=\overline{c^{0}}$, we solve numerically the IBVP 
\begin{subequations}\label{testeq}
\begin{alignat}{3}
&2\partial_tu(t,x)=\frac{1}{9}\partial_x\left((D_0+\delta x)\partial_xu(t,x)\right)\nonumber\\
&\phantom{2\partial_tu(t,x)}+\partial_x \left(u(t,x)\left(\theta\frac{1}{9}\delta-\frac{1}{2}\chi_0\int_{(-1,1)\cap(-x,-x+6)}u(t,x+\xi)\sign(\xi)\,d\xi\right)\right)&&\text{ in }
(0,25]\times(0,6)\label{PL_21D},\\
&\text{no-flux boundary conditions }&&\text{ on }(0,25]\!\times\!\{0,6\},\\
&u_0= 5&&\text{ in }\{0\}\times(0,6).
\end{alignat}
\end{subequations}
Our setting here could mimic a hypothetical in vitro experiment that starts with loss of strong cohesion in a piece of epithelium where cells undergo a partial EMT associated with cancer progression. Such cells acquire increased migratory properties while retaining some cell-cell adhesiveness. This scenario leads to enhanced migration, partly coordinated by cell-cell adhesion, and may result in  formation of cell aggregates. In the present experiment, we assume the extracellular matrix (ECM) to be mildly heterogeneous and not very dense. This is comparable with ECM found in healthy epithelium prior to tissue remodelling by cancer cells.  Therefore, a one-dimensional numerical set-up, as discussed below, provides a suitable description of an early stage of cancer invasion under uncomplicated ECM topology. Simplifying spatial complexity and thus reducing the number of key parameters helps elucidate the interplay between myopic diffusion and cell-cell adhesion components of migration at this stage.

All parameter values used are collected in \cref{TablePar}. {Their selection has been made for illustrative purposes and is not guided by any particular application.}
The obtained results are discussed in \cref{SecNRes} and the simulation approach is explained in \cref{SecNAlg}.

\subsection{Results} \label{SecNRes}
\begin{table}[!tbp]
 \centering
 \begin{subtable}[t]{0.45\textwidth}
 \centering
\begin{tabular}{r|l}
final time& $25$\\
$d$& $1$\\
spatial domain & $(0,6)$\\
$r$& $1$\\
$a$& $1$\\
$F(x)$& $1$\\
$D[q](x)$& $D_0+\delta x$\\
$\chi(x)$& $\chi_0$\\
$u_0(x)$&$5$
\end{tabular}
\caption{Fixed model parameters.}\label{TablePar1}
\end{subtable}
\hfill
\begin{subtable}[t]{0.45\textwidth}
 \centering
 \begin{tabular}{r|l}
  $D_0$&$\{0.15,3,5,8\}$\\
  $\delta$&$\{0,1\}$\\
  $\chi_0$&$\{0.5,1\}$\\
 {$\theta$} &  {$\{0,1\}$}
 \end{tabular}
\caption{Varying model  parameters.}
 \label{TablePar2}
\end{subtable}
\\[1em]
\begin{subtable}[t]{\textwidth}
 \centering
 \begin{tabular}{r|l}
 size of time integration interval & $0.01$\\
 size of time mesh for \texttt{pdepe}&$0.01\cdot \frac{1}{41}$\\ 
 size of spatial mesh for \texttt{pdepe}& $0.01$
 \end{tabular}
 \caption{Mesh sizes.}\label{ParMesh}
\end{subtable}
\caption{Parameters used in simulations of \cref{PL_2}.}\label{TablePar}
\end{table}
We singled out the following test parameters: the minimum value of the diffusion coefficient, $D_0$, a diffusion perturbation parameter, $\delta$,  the constant adhesion sensitivity, $\chi_0$, and a binary parameter $\theta$ that if set to be zero renders the diffusion non-myopic. These variables allow us to adjust  the magnitudes of three distinct flux components: the canonical diffusion along the spatial density gradient with diffusion coefficient $\frac{1}{2\cdot9}(D_0+\delta x)$, the advection to the left with a constant speed $\theta\frac{1}{2\cdot 9}\delta$, and the non-local adhesion with strength $\frac{1}{2\cdot 2}\chi_0$ and the distance-independent adhesion force with $F\equiv1$. The goal of our numerical study here is to understand how the combination of the three motion effects can unfold. 

Since initially the  population is homogeneously distributed over the spatial domain, it would have remained so in the absence of adhesion.
In our simulations, the adhesion strength $\chi_0$ is strictly positive, so that adhesion is always present and, as expected, it promotes the formation of aggregates. 

In our first series of simulations we take $\theta=1$, meaning that diffusion is myopic.
The scenario involving a constant diffusion coefficient, i.e. $\delta=0$, was previously examined in \cite{ButtenHillen}, and we regard it as the control case. The results are shown in the first two columns of \cref{FigureSim1}. In addition to a single aggregate or a pair of aggregates displayed in \cite{ButtenHillen} we also found, for a weak  diffusion corresponding to  $D_0=0.15$ that three very tight aggregates can be observed for $\chi_0=0.5$ in \cref{SF10.150.50}. For $\chi_0=1$, they are about twice as dense, see \cref{SF10.1510}. Comparing the plots in the first column of \cref{FigureSim1}, we see that  substantially increasing the diffusion coefficient  diminishes the density of the aggregates, slows down their formation, and reduces their number. The latter can occur early on or later in time, as, e.g.  \cref{SF1510,SF1310} respectively convey. The aggregates are eventually spaced more than one unit apart, as could be expected for $r=1$. Considering that the adhesion effect increases with growing density inside the sensing region, the impact of a density-independent diffusion is most noticeable during the initial accumulation phase. If diffusion brings density accumulations close enough, adhesion will cause them to merge, as shown, e.g. in  \cref{SF180.50}.

Next, we introduce a linear perturbation into the diffusion coefficient by taking $\delta=1$. The results are shown in the last two columns of \cref{FigureSim1}. This symmetry-breaking effect both increases diffusion, most notably on the right half of the spatial domain, and adds advection to the left. Already for small $D=0.15$, three aggregates can no longer be sustained, with the accumulations in the middle and to the right merging into a single aggregate, see, e.g.  \cref{SF10.1511}. Depending on $D_0$ and $\chi_0$, we observe either two aggregates with the left one now being tighter than the right one, compare, e.g. \cref{SF1511} with \cref{SF1510}, or a single aggregate emerging sooner or later in the left half of the domain, as can be observed, e.g. in \cref{SF180.51}.

Now we set $\theta=0$ and $\delta=1$. This eliminates the constant-speed advection to the left, rendering the diffusion non-myopic while still retaining the spatial heterogeneity in the diffusion coefficient. The numerical results for this case are presented in the odd columns of \cref{FigureSim2} next to the corresponding graphs for $\theta=1$ in the even columns of the same Figure.  The comparison of the plots reveals that advection due to diffusion being myopic primarily serves to preserve the position of the aggregates. Only for $D_0=8$ and $\chi_0=0.5$ it is different, see \cref{SF280.51,SF380.51}. Here in the absence of the density-independent advection in the left direction a tight aggregate fails to form. This might be because the density accumulation fails to reach the left half of the domain and hence stays under the influence of a  strong diffusion that precludes a substantial aggregation.

\begin{figure}%
\foreach \Diff[count=\i] in {8,5,3,0.15}{%
\centering
\adjustbox{raise=1.8cm}{\begin{minipage}{0.01\textwidth}
\rotatebox{90}{$D_0=$\Diff}
\end{minipage} } 
\foreach \T in {25}{%
\foreach \delt in {0,1}{%
\foreach \chiz in {0.5,1}{%
\IfFileExists{MATLAB/1D_Myopic/T=\T,D=\Diff,delta=\delt,chi0=\chiz.jpeg}{%
\begin{subcaptionblock}{0.23\textwidth}
\centering
\begin{overpic}[width=1.1\textwidth]{MATLAB/1D_Myopic/T=\T,D=\Diff,delta=\delt,chi0=\chiz.jpeg}
\put(20,80){\ifnum \i=1 $(\chi_0,\delta)=(\chiz,\delt)$\fi}
\end{overpic}
\caption{}
\label{SF1\Diff\chiz\delt}
\end{subcaptionblock}
  }{%
  }
}}
}
\\
}
\caption{Kymographs of numerical solutions of \cref{testeq} for $\theta=1$ and  various values of $D_0,\chi_0$, and $\delta$.}             
\label{FigureSim1}
\end{figure}

\begin{figure}%
\foreach \Diff[count=\i] in {8,5,3,0.15}{%
\centering
\adjustbox{raise=1.8cm}{\begin{minipage}{0.01\textwidth}
\rotatebox{90}{$D_0=$\Diff}
\end{minipage} } 
\foreach \T in {25}{%
\foreach \chiz in {0.5,1}{%
\IfFileExists{MATLAB/1D_Myopic/T=\T,D=\Diff,delta=1,chi0=\chiz.jpeg}{%
\begin{subcaptionblock}{0.23\textwidth}
\centering
\begin{overpic}[width=1.1\textwidth]{MATLAB/1D_Myopic/T=\T,D=\Diff,delta=1,chi0=\chiz.jpeg}
\put(20,80){\ifnum \i=1 $(\chi_0,\theta)=$(\chiz,1) \fi}
\end{overpic}
  \caption{}
\label{SF2\Diff\chiz1}
\end{subcaptionblock}
  }{%
  }
\IfFileExists{MATLAB/1D_NonMyopic/T=\T,D=\Diff,delta=1,chi0=\chiz.jpeg}{%
\begin{subcaptionblock}{0.23\textwidth}
\centering
\begin{overpic}[width=1.1\textwidth]{MATLAB/1D_NonMyopic/T=\T,D=\Diff,delta=1,chi0=\chiz.jpeg}
\put(10,80){\ifnum \i=1 $(\chi_0,\theta)=$(\chiz,0)  \fi}
\end{overpic}
  \caption{}
\label{SF3\Diff\chiz1}
\end{subcaptionblock}
  }{%
  }
}
}
\\
}
\caption{Kymographs of numerical solutions to \cref{testeq} for $\theta=1$ (odd columns) and $\theta=0$ (even columns), $\delta=1$, and various values of $D_0$ and $\chi_0$.}              
\label{FigureSim2}
\end{figure}

\subsection{Method}\label{SecNAlg}
Our solver for \cref{testeq} was written in MATLAB \cite{MATLAB}. 
In this Subsection, we briefly describe our numerical scheme. To start, we decomposed the time interval $[0,25]$ into intervals of equal size of $0.01$. Let $t_k$, $k\in\{0,1,\dots,25\times 100\}$, be the  corresponding mesh points. Starting with $\tilde u_0:=u_0$, we successively determined  approximations $\tilde u_k$ of $u(t_k,\cdot)$ by solving numerically the IBVP
\begin{subequations}\label{deqeq}
\begin{alignat}{3}
&2\partial_t\tilde u(t,x)=\frac{1}{9}\partial_x\left((D_0+\delta x)\partial_x\tilde u(t,x)\right)\nonumber\\
&\phantom{2\partial_t\tilde u(t,x)}\!\!\!\!+\partial_x \left(\tilde u(t,x)\left(\theta\frac{1}{9}\delta-\frac{1}{2}\chi_0\int_{(-1,1)\cap(-x,-x+6)}\tilde u_k(x+\xi)\sign(\xi)\,d\xi\right)\right)&&\text{ in }
\left(0,0.01\right]\times(0,6)\label{PDE1},\\
&\text{no-flux boundary conditions }&&\text{ on }\left(0,0.01\right]\!\times\!\{0,6\},\\
&\tilde u= \tilde u_k&&\text{ in }\{0\}\times(0,6),
\end{alignat}
\end{subequations}
and then set
\begin{align*}
 \tilde u_{k+1}:=\tilde u(0.01,\cdot).
\end{align*}
To solve the implicit-explicit equation  \cref{deqeq}, we used MATLAB's function \texttt{pdepe} \cite{pdepe}.
 The sizes of the equidistant  meshes that we passed to \texttt{pdepe} were $1/41\cdot 0.01$ for time and $0.01$ for space. To apply \texttt{pdepe} in this case, one needs first to produce the function described by the non-local term, as it is part of the expression for the flux. We  discretised that integral term for the entire spatial domain    replacing it by the right Riemann sums
\begin{align}
 &s_m:=0.01\sum_{l=a(m)}^{b(m)}\tilde u_k\left(0.01(m+l)\right)\sign\left(0.01l\right),\label{Rsum}
 \end{align}
 where
 \begin{align}
 &a(m):=\max\{-99,-m\},\quad  b(m):=\min\{100,-m+100\cdot 6\},\qquad m\in\{0,\dots,100\cdot 6\},\nonumber 
\end{align}
and then interpolated these values  using MATLAB's function \texttt{griddedInterpolant} \cite{interp} to create a function of $x$. To compute the sums \cref{Rsum}, we applied MATLAB's function \texttt{xcorr} \cite{xcorr} as follows. Let 
\begin{align*}
 &\omega:=(\underbrace{-1,\dots,-1}_{(100-1)\text{ times}},0,\underbrace{1,\dots,1}_{100\text{ times}}),\\
 &u:=(\tilde u_k(0),\tilde u_k(0.01),\dots, \tilde u_k(6)),\\
 &w:=\texttt{xcorr}(u,\omega),
\end{align*}
then 
\begin{align*}
 (s_0,\dots,s_{100\cdot 6})=0.01\left(w_{(100\cdot 6+1)-(100-1)},\dots,w_{2(100\cdot 6+1)-100}\right).
\end{align*}
\section{Discussion}
\label{SecDisc}
CCA plays a pivotal role in the development and functioning of multicellular organisms. Notably, it regulates cell migration, either promoting or inhibiting it. 
Macroscopic mathematical models can contribute to a better understanding of adhesion effects because they are amenable to both rigorous mathematical analysis and {\it in silico} studies, and numerical results for these models can be compared to medical images.

In this paper we aimed at devising a new multiscale approach to modelling cell migration driven by such effects as CCA and anisotropic diffusion.  After reviewing  previously available approaches in \cref{SecIntroMod},  we derived in \cref{SecBasis}  classes of  IPDE models  containing non-local adhesion and myopic diffusion,  the classical model in \cite{Armstrong2006} being their special case. We further extended our approach in \cref{SecCAMs} where we derived a novel model \cref{MacroModel} that can account for subcellular binding dynamics of CAMs, molecules responsible for cells sticking to each other.

Our modelling can serve as a  starting point for considerably more realistic models for adhesion-driven motion. 
One of our simplifying modelling assumptions, \cref{Ass1} in \cref{SecCAMs}, was that all junctions are cell-cell junctions of one and the same type. The main application that we had in mind here were the adherens junctions such as facilitated by E-cadherins, see \cref{IntroBB}. In reality, cells form a variety of junctions and, moreover, subpopulations with distinct adhesion properties can be involved. Our approach could be extended to accommodate such complexities.

The usefulness of such models as \cref{MacroModel} hinges on their solvability. As previously announced in \cref{RemSolvMM}, we prove local solvability, yet for a quasilinear non-myopic degenerate diffusion, in a coming paper. We postpone to that work a discussion of challenges that arise in connection with treatment of equation \cref{momy01}.
While non-myopic diffusion is often adopted in modelling cell motion, our simulation results in \cref{SecNum} underscore the difference that a myopic diffusion can make.  
We are going to settle the solvability for \cref{MacroModel} in a future work.
\section*{Acknowledgement}
\addcontentsline{toc}{section}{Acknowledgement}
\begin{itemize}
\item The authors thank Christina Surulescu {(RPTU Kaiserslautern-Landau)} for stimulating discussions.
\item 
The authors were supported by the Engineering and Physical Sciences Research Council [grant number
EP/T03131X/1].
\item For the purpose of open access, the authors have applied a Creative Commons Attribution (CC BY) licence to any Author Accepted Manuscript version arising.
\item All data is provided in full in \cref{SecNum} of this paper.                                                            \end{itemize}

\phantomsection
\addcontentsline{toc}{section}{References}
\printbibliography

\begin{appendices}
\section{}\label{SecApp}
 The following Lemma extends the result that is well-known for skew-symmetric kernels $K$, see, e.g. {\cite[Chapter 3 \S2 Theorem 3.2.1(b)]{Golse}}. The proof, which is very similar, is provided for the reader's convenience.  
\begin{Lemma}\label{LemA1}
 Let $m,N\in\N$. Let $A\in C([0,\infty)\times\R^m;\R^m)$ and $K\in C([0,\infty)\times\R^m\times \R^m; \R^m)$. Define 
 \begin{align}
 &\K\mu(t,Z):=\int_{\R^m}K(t,Z,Z')\,\mu(dZ'),\qquad \tr K(t,Z):=K(t,Z,Z),\qquad Z\in\R^m.\label{KKern}
\end{align}
Consider the ODE system
 \begin{align}
  \frac{dZ_i}{dt}=\frac{1}{N}\sum_{\underset{j\neq i}{j=1}}^NK(\cdot,Z_i,Z_j)+A(\cdot,Z_i).\label{Veq}
 \end{align}
Let $(Z_1^T,\dots,Z_N^T)^T\in C^1([0,T);\R^{mN})$ be a solution to \cref{Veq} for some $T\in[0,\infty]$. Then, the empirical measure
\begin{align*}
 \mu_N:=\frac{1}{N}\sum_{i=1}^N\delta_{Z_i}
\end{align*}
satisfies 
\begin{align}
 \partial_t \mu_N+\nabla_z\cdot \left(\left(\K\mu_N-\frac{1}{N}\tr K+A\right)\mu_N\right)=0\label{MeanFieled}
\end{align}
on $(0,T)\times \R^m$ in the distributional sense.

\end{Lemma}
\begin{proof}
Throughout the proof, equations are understood to be satisfied in the sense of distributions on $(0,T)\times \R^m$.

Let $i\in\{1,\dots,N\}$. From the theory of CTEs it is known that
  (see, e.g. \cite{Golse}) that  
  \begin{align}
   \partial_t \delta_{Z_i}+\nabla_z\cdot\left(\left(\frac{1}{N}\sum_{\underset{j\neq i}{j=1}}^NK(\cdot,\cdot,Z_j)+A\right)\delta_{Z_i}\right)=0.\label{CTi_}
  \end{align}
   Using \cref{KKern} and the property $$\varphi\delta_{\bar Z}=\varphi(\bar Z)\delta_{\bar Z}$$ which holds for any $\bar Z\in\R^m$ and continuous function $\varphi=\varphi(Z)$, we can rewrite \cref{CTi_} as follows:
  \begin{align}
  0=&\partial_t \delta_{Z_i}+\nabla_z\cdot\left(\left(\frac{1}{N}\sum_{\underset{j\neq i}{j=1}}^NK(\cdot,\cdot,Z_j)+A\right)\delta_{Z_i}\right)\nonumber\\
=&\partial_t \delta_{Z_i}+\nabla_z\cdot\left(\left(\frac{1}{N}\sum_{{j=1}}^N\K \delta_{Z_j}-\frac{1}{N}\tr K+A\right)\delta_{Z_i}\right)\nonumber\\
=&\partial_t \delta_{Z_i}+\nabla_z\cdot\left(\left(\K \mu_N-\frac{1}{N}\tr K+A\right)\delta_{Z_i}\right).
   \label{CTi}
  \end{align}
 Computing the arithmetic  mean with respect to $i$ on both sides of \cref{CTi} yields \cref{MeanFieled}.
\end{proof}

\end{appendices}

\end{document}